\documentclass[11pt]{article}
\usepackage{fullpage}
\usepackage{url}
\usepackage{xspace}

\usepackage{graphics}
\usepackage[dvips]{epsfig}

\usepackage{amsmath}
\usepackage{amssymb}
\usepackage{amsfonts}
\usepackage{graphicx}
\usepackage{algorithm, algorithmic}
\usepackage{comment}

\newtheorem{theorem}{Theorem}[section]
\newtheorem{lemma}{Lemma}[section]
\newtheorem{corollary}{Corollary}[section]
\newtheorem{claim}{Claim}[section]

\newtheorem{definition}{Definition}[section]

\newcommand{\qed}{\hfill $\Box$ \bigbreak}
\newenvironment{proof}{\noindent {\bf Proof.}}{\qed}

\newcommand{\remove}[1]{}

\newcommand{\algname}{\ensuremath{{\cal A}}}
\newcommand{\diam}[1]{\ensuremath{{diam}(#1)}}
\newcommand{\chop}[1]{\ensuremath{\mathrm{chop}(#1)}}
\newcommand{\selected}[1]{\ensuremath{\mathrm{max}_{\cal A}(#1)}}
\newcommand{\advice}[1]{\ensuremath{\mathrm{adv}_{\algname}(#1)}}
\newcommand{\glue}[2]{\ensuremath{\mathrm{glue}(#1,#2)}}

\newcommand{\candidates}[1]{\ensuremath{C_{#1}}}
\newcommand{\maxcandidates}{\ensuremath{\beta}}
\newcommand{\maxcandidatesval}{\ensuremath{\lceil8/\alpha\rceil}}

\newcommand{\givensize}{\ensuremath{\tilde{n}}}
\newcommand{\givendiameter}{\ensuremath{D}}

\begin{document}

\baselineskip  0.18in 
\parskip     0.0in 
\parindent   0.2in 

\title{{\bf Election vs. Selection:\\ Two Ways of Finding the Largest Node in a Graph}}
\date{}
\newcommand{\inst}[1]{$^{#1}$}

\author{
Avery Miller\inst{1},
Andrzej Pelc\inst{1}$^,$\footnote{Partially supported by NSERC discovery grant and by the Research Chair in Distributed Computing at the Universit\'e du Qu\'{e}bec en Outaouais.}\\
\inst{1} Universit\'{e} du Qu\'{e}bec en Outaouais, Gatineau, Canada.\\
E-mails: \url{avery@averymiller.ca}, \url{pelc@uqo.ca}\\
}

\date{ }
\maketitle

\begin{abstract}

Finding the node with the largest label in a labeled network, modeled as an undirected connected graph, is one of the fundamental problems in distributed computing.
This is the way in which {\em leader election} is usually solved. We consider two distinct tasks in which the largest-labeled node is found deterministically.
In {\em selection}, this node has to output 1 and all other nodes have to output 0. In {\em election}, the other nodes must additionally learn the largest
label (everybody has to know who is the elected leader). Our aim is to compare the difficulty of these two seemingly similar tasks executed under stringent running time constraints. The measure of difficulty is the amount of information
 that nodes of the network must initially possess, in order to solve the given task in an imposed amount of time. Following the standard framework of {\em algorithms
with advice}, this information (a single binary string) is provided to all nodes at the start by an oracle knowing the entire graph. The length of this string is called the {\em size of advice}. The paradigm of algorithms with advice has a far-reaching importance in the realm of network algorithms. Lower bounds on the size of advice
give us impossibility results based strictly on the \emph{amount} of initial knowledge outlined in a model's description.
This more general approach should be contrasted with
traditional results that focus on specific \emph{kinds} of information available to nodes, such as the size, diameter, or maximum node degree. 

Consider the class of $n$-node graphs with any diameter $diam \leq D$, for some integer $D$.
If time is larger than $diam$, then both tasks can be solved without advice.
For the task of {\em election}, we show that if time is smaller than $diam$, then the optimal size of advice is $\Theta(\log n)$,
and if time is exactly $diam$, then the optimal size of advice is $\Theta(\log D)$.  For the task of {\em selection}, the situation changes dramatically,
even within the class of rings. Indeed, for the class of rings, we show that, if time is
$O(diam^{\epsilon})$, for any $\epsilon <1$, then the optimal size of advice is $\Theta(\log D)$, and, if time is $\Theta(diam)$ (and at most $diam$) then
 this optimal size is $\Theta(\log \log D)$. Thus there is an {\em exponential} increase of difficulty (measured by the size of advice) between selection in time $O(diam^{\epsilon})$, for any $\epsilon <1$,
and selection in time $\Theta(diam)$. As for the comparison between election and selection, our results show that, perhaps surprisingly, while  for small time, the difficulty of these two tasks 
 on rings is similar, for time  $\Theta(diam)$ the difficulty of election (measured by the size of advice) is exponentially larger than that of selection.
\vspace{2ex}

\end{abstract}

\vfill

\vfill

\thispagestyle{empty}
\setcounter{page}{0}
\pagebreak

\section{Introduction}

{\bf Background.} Finding the node with the largest label in a labeled network is one of the fundamental problems in distributed computing.
This is the way in which {\em leader election} is usually solved. (In leader election, one node of a network has to become a {\em leader}
and all other nodes have to become {\em non-leaders}). In fact, to the best of our knowledge, all existing leader election algorithms
performed in labeled networks choose, as leader, the node with the largest label or the node with the smallest label \cite{Ly}. The classic problem
of leader election first appeared in the study of local area token ring networks \cite{LL}, where, at all times, exactly one node (the owner of a circulating token) has the right to initiate
communication. When the token is accidentally lost, a leader is elected as the initial owner of the token.

\noindent
{\bf Model and Problem Description.} The network is modeled as an undirected connected graph with $n$ labeled nodes and with diameter $diam$ at most $D$.
We denote by $\diam{G}$ the diameter of graph $G$.
Labels are drawn from the set of integers $\{1,\dots,L\}$, where $L$ is polynomial in $n$. Each node has a distinct label.
Initially each node knows its label and its degree. The node with the largest label in a graph will be called its {\em largest node}.

We use the extensively studied $\cal{LOCAL}$ communication model \cite{Pe}. In this model, communication proceeds in synchronous
rounds and all nodes start simultaneously. In each round, each node
can exchange arbitrary messages with all of its neighbours and perform arbitrary local computations. For any $r \geq 0$ and any node $v$, we use
$K(r,v)$ to denote     
the knowledge acquired by $v$ within $r$ rounds. Thus, $K(r,v)$  consists of the subgraph induced by all nodes at distance at most $r$
from $v$, except for the edges joining nodes at distance exactly $r$ from $v$, and of degrees (in the entire graph) of all nodes at distance exactly $r$ from $v$. Hence,
if no additional knowledge is provided {\em a priori} to the nodes, the decisions of $v$ in round $r$ in any deterministic algorithm are a function of $K(r,v)$. We denote by $\lambda(r,v)$ the set of labels of nodes in the subgraph induced by all nodes at distance at most $r$
from $v$.
The {\em time} of a task is the minimum number of rounds sufficient to complete it by all nodes. 

It is well known that the synchronous process of the $\cal{LOCAL}$  model can be simulated in an asynchronous network. This can be achieved 
by defining for each node separately its asynchronous round $i$;
in this round, a node performs local computations, then sends messages stamped $i$ to all neighbours, and  waits until it gets messages stamped $i$ from all neighbours.
To make this work, every node is required to send at least one (possibly empty) message with each stamp, until termination.
All of our results can be translated for asynchronous networks by replacing ``time of completing a task''  by ``the maximum number of asynchronous rounds  to complete it, taken over all nodes''.

We consider two distinct tasks in which the largest node is found deterministically.
In {\em selection}, this node has to output 1 and all other nodes have to output 0. In {\em election}, all nodes must output the largest label.
Note that in election all nodes perform selection and additionally learn the identity of the largest node.
Both variations are useful in different applications. In the aforementioned application of recovering a lost
token, selection is enough, as the chosen node will be the only one to get a single token and then the token will be passed from node to node. In this case,
other nodes do not need to know the identity of the chosen leader. The situation is different if all nodes must agree on a label of one of the nodes, e.g., to use it later
as a common parameter for further computations. Then the full strength of election is needed. Likewise, learning the largest label by all nodes is important when
labels carry some additional information apart from the identities of nodes, e.g., some values obtained by sensors located in nodes. Our results also remain valid
in such situations, as long as the  ``informative label'' can be represented as an integer polynomial in $n$.  

Our aim is to compare the difficulty of the two seemingly similar tasks of selection and election executed under stringent running time constraints. 
The measure of difficulty is the amount of information
 that nodes of the network must initially possess in order to solve the given task in an imposed amount of time. Following the standard framework of {\em algorithms
with advice}, see, e.g.,   \cite{DP,EFKR,FGIP,FKL,FP,IKP,SN}, this information (a single binary string) is provided to all nodes at the start by an oracle knowing the entire graph. The length of this string is called the {\em size of advice}. 

The paradigm of algorithms with advice has a far-reaching importance in the realm of network algorithms. Establishing a tight bound on the minimum size of advice sufficient to accomplish a given task permits to rule out
entire classes of algorithms and thus focus only on possible candidates. For example, if we prove that $\Theta(\log n)$ bits of advice are needed to perform a certain task in $n$-node graphs, this rules out all 
potential algorithms that can work using only some linear upper bound on the size of the network, as such an upper bound could be given
to the nodes using $\Theta(\log \log n)$ bits by providing them with
$\lceil \log n \rceil$. Lower bounds on the size of advice
give us impossibility results based strictly on the \emph{amount} of initial knowledge outlined in a model's description.
This more general approach should be contrasted with
traditional results that focus on specific \emph{kinds} of information available to nodes, such as the size, diameter, or maximum node degree.

\noindent
{\bf Our results.} Consider the class of $n$-node graphs with any diameter $diam \leq D$, for some integer $D$.
First observe that if time is larger than $diam$, then both tasks can be solved without advice,
as nodes learn the entire network and {\em learn that they have learned it}. Thus they can just choose the maximum of all labels seen.

For the task of {\em election}, we show that if time is smaller than $diam$, then the optimal size of advice is $\Theta(\log n)$,
and if time is exactly $diam$, then the optimal size of advice is $\Theta(\log D)$. 
Here our contribution consists in proving two lower bounds on the size of advice. We prove one lower bound by exhibiting, for any positive integers $D<n$, networks of size $\Theta(n)$ and diameter $\Theta(D)$, for which
$\Omega(\log n)$ bits of advice are needed for election in time below $diam$. To prove the other lower bound, we present, for any positive integer $D$, networks of diameter 
$diam \leq D$ for which $\Omega(\log D)$ bits of advice are needed for election in time exactly $diam$. 
These lower bounds are clearly tight, as even for time 0, $O(\log n)$ bits of advice are enough to provide the largest label in the network, and, for time $diam$ --
when nodes know the entire network, but they {\em do not know that they know it} -- $O(\log D)$ bits of advice are enough to give the diameter $diam$ to all nodes and thus reassure them that they have seen everything. In this case,  they can safely choose the maximum of all labels seen.

Hence, a high-level statement of our results for election is the following. If time is too small for all nodes to see everything, then no more efficient help in election is possible
than just giving the largest label. If time is sufficient for all nodes to see everything, but too small for them to realize that they do, i.e., the time is exactly $diam$, then no more efficient help in election is possible than providing $diam$, which supplies nodes with the missing certainty that they have seen everything.

It should be noted that $n$ could be exponential
in $D$, as in hypercubes, or $D$ can even be constant with respect to arbitrarily large  $n$. Thus, our results for election show that, for some networks, there is an exponential (or even
larger) gap of difficulty of election (measured by the size of advice) between time smaller than $diam$ and time exactly $diam$. Another such huge gap is between time
$diam$, when advice of size $\Theta(\log D)$ is optimal, and time larger than $diam$, when 0 advice is enough. These gaps could be called {\em intra-task}
jumps in difficulty for election with respect to time.

For the task of {\em selection}, the situation changes dramatically,
even within the class of rings. Indeed, for the class of rings, we show that, if time is
$O(diam^{\epsilon})$, for any $\epsilon <1$, then the optimal size of advice is $\Theta(\log D)$, and, if time is $\Theta(diam)$ (and at most $diam$) then
 this optimal size is $\Theta(\log \log D)$. Here our contribution is three-fold. For selection in time $O(diam^{\epsilon})$, for any $\epsilon <1$,  we exhibit, for any
 positive integer $D$, a class of rings with diameter at most $D$ which requires advice of size $\Omega(\log D)$. As before, this lower bound is tight, even for time 0.
Further, for selection in time at most $\alpha \cdot diam$, for $\alpha \leq 1$, we construct a class of rings with diameter $diam \leq D$ which requires advice of size $\Omega(\log\log D)$. At first glance, it might seem that this lower bound is too weak. Indeed, providing  either the diameter or the largest label, which are both natural choices of advice, would not give a tight upper bound, as this may require $ \Theta(\log D)$ bits. However, we use a more sophisticated idea that permits us to construct a very compact advice 
(of matching size $O(\log\log D)$) and we design a selection algorithm,
working for all rings of diameter $diam \leq D$ in time at most $\alpha \cdot diam$, for which this advice is enough.

 Thus there is an {\em exponential} increase of difficulty (measured by the size of advice) between selection in time $O(diam^{\epsilon})$, for any $\epsilon <1$,
and selection in time $\Theta(diam)$. As in the case of election, another huge increase of difficulty occurs between time $diam$ and time larger than $diam$.  These gaps could be called {\em intra-task} jumps of difficulty for selection with respect to time. 

 As for the comparison between election and selection, our results show that, perhaps surprisingly, while  for small time, the difficulty of these two tasks 
 on rings is similar, for time  $\Theta(diam)$ the difficulty of election (measured by the size of advice) is exponentially larger than that of selection,
even for the class of rings.
 While both in selection and in election the unique leader having the maximum label is chosen, these tasks differ in how widely this label is known. It follows from our results that,
if linear time (not larger than the diameter) is available, then the increase of difficulty (in terms of advice) of making this knowledge widely known is exponential.
This could be called the {\em inter-task} jump of difficulty between election and selection.  Figure \ref{summarytable} provides a summary of our results.

\begin{figure}[!ht]
\begin{center}
\includegraphics[scale=1]{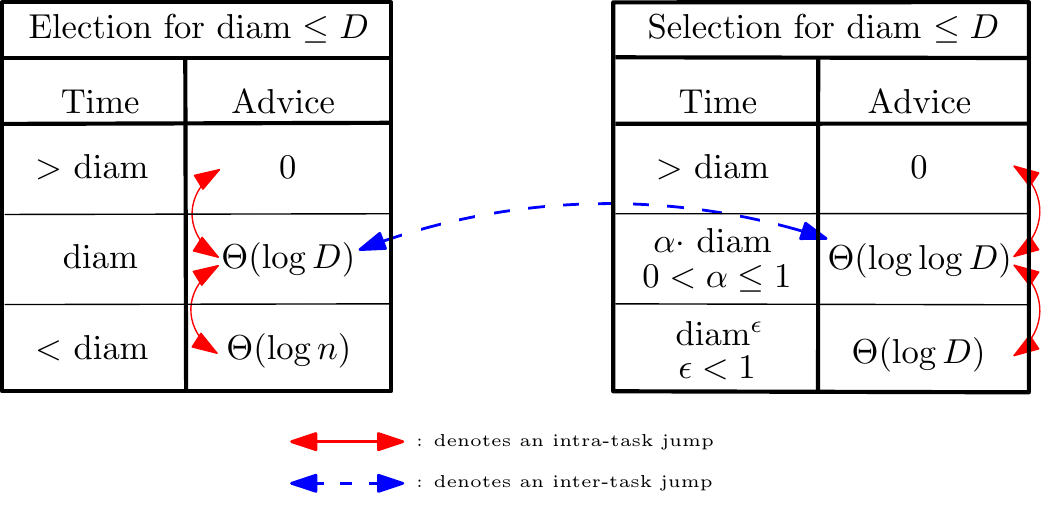}
\end{center}
\caption{Tight bounds on the size of advice for election and selection in arbitrary graphs and rings (respectively) with diameter diam $\leq D$}
\label{summarytable}
\end{figure}


\noindent
{\bf Related work.}
The leader election problem was introduced in \cite{LL}. This problem  has been extensively studied in the scenario adopted in the present paper, i.e.,
where all nodes have distinct labels. As far as we know, this task was always solved by finding either the node with the largest or that with the smallest label.
Leader election was first studied for rings.
A synchronous algorithm based on label comparisons and using
$O(n \log n)$ messages was given in \cite{HS}. It was proved in \cite{FL} that
this complexity is optimal for comparison-based algorithms. On the other hand, the authors showed
an algorithm using a linear number of messages but requiring very large running time.
An asynchronous algorithm using $O(n \log n)$ messages was given, e.g., in \cite{P}, and
the optimality of this message complexity was shown in \cite{B}. Deterministic leader election in radio networks has been studied, e.g., 
in \cite{JKZ,KP,NO}, as well as randomized leader election, e.g., in \cite{Wil}. In \cite{HKMMJ}, the leader election problem was
approached in a model based on mobile agents for networks with labeled nodes.

Many authors \cite{An,AtSn,ASW,BV,YK2,YK3} studied leader election
in anonymous networks. In particular, \cite{BSVCGS,YK3} characterize message-passing networks in which
leader election can be achieved when nodes are anonymous.
Characterizations of feasible instances for leader election were provided in~\cite{C,CM}.
Memory needed for leader election in unlabeled networks was studied in \cite{FP}.

Providing nodes or agents with arbitrary kinds of information that can be used to perform network tasks more efficiently has previously been
proposed in \cite{AKM01,CFP,DP,EFKR,FGIP,FIP1,FIP2,FKL,FP,FPR,GPPR02,IKP,KKKP02,KKP05,SN,TZ05}. This approach was referred to as
{\em algorithms with advice}.  
The advice is given either to nodes of the network or to mobile agents performing some network task.
In the first case, instead of advice, the term {\em informative labeling schemes} is sometimes used, if (unlike in our scenario) different nodes can get different information.

Several authors studied the minimum size of advice required to solve
network problems in an efficient way. 
 In \cite{KKP05}, given a distributed representation of a solution for a problem,
the authors investigated the number of bits of communication needed to verify the legality of the represented solution.
In \cite{FIP1}, the authors compared the minimum size of advice required to
solve two information dissemination problems using a linear number of messages. 
In \cite{FKL}, it was shown that advice of constant size given to the nodes enables the distributed construction of a minimum
spanning tree in logarithmic time. 
In \cite{EFKR}, the advice paradigm was used for online problems.
In \cite{FGIP}, the authors established lower bounds on the size of advice 
needed to beat time $\Theta(\log^*n)$
for 3-coloring cycles and to achieve time $\Theta(\log^*n)$ for 3-coloring unoriented trees.  
In the case of \cite{SN}, the issue was not efficiency but feasibility: it
was shown that $\Theta(n\log n)$ is the minimum size of advice
required to perform monotone connected graph clearing.
In \cite{IKP}, the authors studied radio networks for
which it is possible to perform centralized broadcasting in constant time. They proved that constant time is achievable with
$O(n)$ bits of advice in such networks, while
$o(n)$ bits are not enough. In \cite{FPR}, the authors studied the problem of topology recognition with advice given to nodes. 
To the best of our knowledge, the problems of leader election or maximum finding with advice have never been studied before.

\section{Election}

Notice that in order to perform election in a graph $G$ in time larger than its diameter, no advice is needed. Indeed, after time $\diam{G}+1$, all nodes
know the labels of all other nodes, and  they are aware that they have this knowledge. This is because, in round $\diam{G}+1$, no messages containing new labels are received by any node. So, it suffices for all nodes to output the largest of the labels that they have seen up until round $\diam{G}+1$. 

Our first result shows that, if election time is no more than the diameter of the graph, then the size of advice must be at least logarithmic in the diameter. 
This demonstrates a dramatic difference between the difficulty of election in time $\diam{G}$ and election in time $\diam{G}+1$, measured by the minimum size of advice.

\begin{theorem}\label{diam}
Consider any algorithm ${\cal A}$ such that, for every graph $G$, algorithm ${\cal A}$ solves election within $\mathrm{diam}(G)$ rounds. For every integer $D \geq 2$, there exists a ring of diameter at most $D$ for which algorithm ${\cal A}$ requires advice of size $\Omega(\log D)$.
\end{theorem}
\begin{proof}
Fix any integer $D \geq 2$. We will show a stronger statement: at least $D$ different advice strings are needed in order to solve election within $\diam{G}$ rounds for 
rings $G$ with diameter at most $D$. The high-level idea of the proof is to first construct a particular sequence of $D$ rings of increasing sizes, each with a different largest label.
With few advice strings, two such rings get the same advice. We show that there is a node in these two rings which acquires the same knowledge when executing algorithm $\cal A$, and hence has to elect the same node in both rings, which is a contradiction. 

To obtain a contradiction, assume that $D-1$ different advice strings suffice. Consider a ring $R$ of diameter $D$ whose node labels form the sequence $(1,2,\ldots,D+1,2D+1,2D,\ldots,D+2)$ (see Figure \ref{rings}). For each $k \in \{1,\ldots,D\}$, define $R^1_{k}$ to be the ring obtained from $R$ by taking the subgraph induced by the nodes at distance at most $k$ from node $1$ and adding an edge between nodes $k+1$ and $D+k+1$ (see Figure \ref{rings}). First, note that, for each $k \in \{1,\ldots,D\}$, the diameter of $R^1_k$ is $k$, and, the largest node in $R^1_k$ has label $D+k+1$. The correctness of ${\cal A}$ implies that, when ${\cal A}$ is executed at a node in $R^1_k$, it must halt within $k$ rounds and output $D+k+1$. Next, by the Pigeonhole Principle, there exist $i,j \in \{1,\ldots,D\}$ with $i < j$ such that the same advice string is provided to nodes of both $R^1_i$ and $R^1_{j}$ when they execute ${\cal A}$. When executed at node 1 in $R^1_i$, algorithm ${\cal A}$ halts in some round $r \leq \diam{R^1_i} = i$ and outputs $D+i+1$. We show that, when executed at node 1 in $R^1_{j}$, algorithm ${\cal A}$ also halts in round $r$ and outputs $D+i+1$. Indeed, the algorithm is provided with the same advice string for both $R^1_i$ and $R^1_{j}$, and, $K(r,1)$ in $R^1_i$ is equal to  $K(r,1)$ in $R^1_{j}$. This contradicts the correctness of ${\cal A}$ since, for the ring $R^1_{j}$, there is a node with label $D+j+1>D+i+1$. To conclude, notice that, 
since at least $D$ different advice strings are needed for the class of rings of diameter at most $D$, the size of advice must be $\Omega(\log D)$ for at least one
of these rings. 
\end{proof}

\begin{figure}[!ht]
\begin{center}
\includegraphics[scale=1]{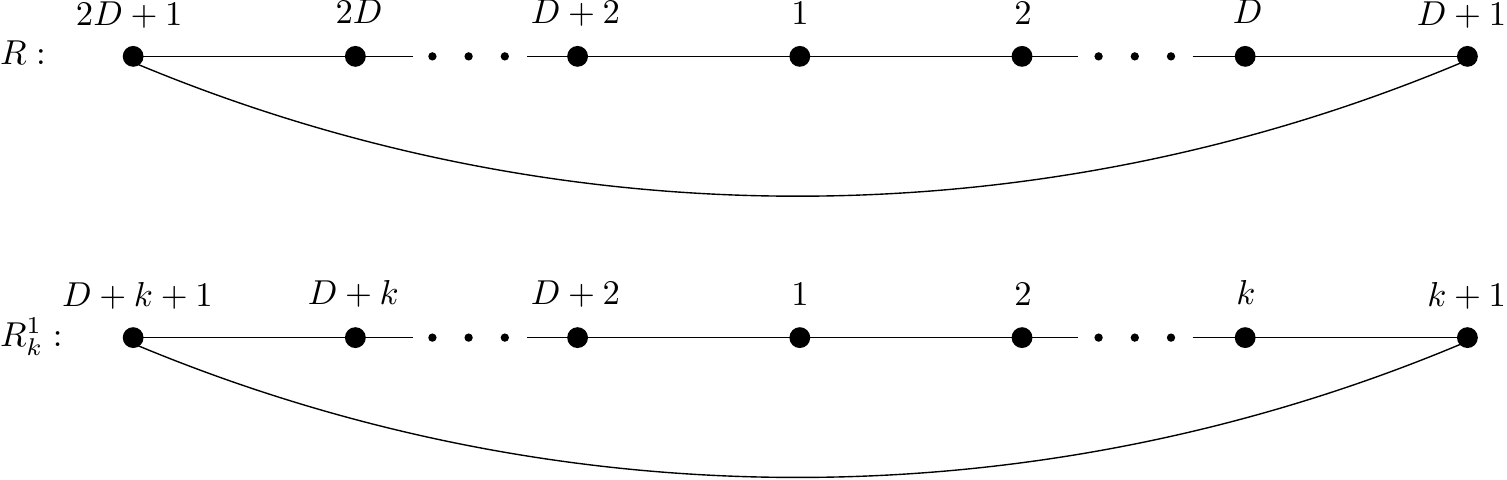}
\end{center}
\caption{Rings $R$ and $R^1_k$, as constructed in the proof of Theorem \ref{diam}}
\label{rings}
\end{figure}

Note that the lower bound established in Theorem \ref{diam} is tight. Indeed, to achieve election in time $\diam{G}$ for any graph $G$, it is enough to provide
the value of $\diam{G}$ to the nodes of the graph and have each node elect the node with the largest label it has seen up until round  $\diam{G}$.
Hence we have the following corollary for the class of graphs of diameter at most $D$.

\begin{corollary}\label{cor1}
The optimal size of advice to complete election in any graph in time at most equal to its diameter is $\Theta(\log D)$.
\end{corollary}

We next consider accomplishing election in time less than the diameter of the graph. One way to do it is to provide the maximum label to all nodes as advice. This yields 
election in time 0 and uses advice of size $O(\log n)$, where $n$ is the size of the graph,  since the space of labels is of size $L$ polynomial in $n$.
The following result shows that this size of advice cannot be improved for election in {\em any} time below the diameter. This result, when compared to Corollary
\ref{cor1}, again shows the dramatic difference in the difficulty of election (measured by the minimum size of advice) but now between times $\diam{G}-1$ and $\diam{G}$.


\begin{theorem}
Consider any positive integers $\givendiameter < \givensize $. There exists $n \in \Theta( \givensize )$ such that, for any election algorithm $\mathcal{A}$ in which every execution halts within $D-1$ rounds, there exists an $n$-node graph of diameter $D$ for which the size of advice needed by $\mathcal{A}$ is $\Omega(\log{n})$.
\end{theorem}
\begin{proof}
The high-level idea of the proof is the following. We first construct a family of ``ring-like'' graphs.
For a given number of advice strings, we obtain a lower bound on the number of such graphs for which the same advice is given.
On the other hand, an upper bound on this number is obtained by exploiting the fact that no node can see the entire graph within $D-1$ rounds. Comparing these bounds gives the desired bound on the size of advice.

Let $n$ be the smallest integer greater than $\givensize$ that is divisible by 2D (and note that $n \in \Theta( \givensize )$).

Consider a family $\mathcal{C}$ of $n^3$ pairwise disjoint sets, each of size $\frac{n}{2D}$. In particular, let 
$\mathcal{C}= \{C_0,\ldots,C_{n^3-1}\}$, where
$C_i=\{\frac{ni}{2D}+1,\ldots,\frac{n(i+1)}{2D}\}$.

We construct a family $\mathcal{G}$ of $n$-node graphs. Each graph in $\mathcal{G}$ is obtained by first choosing an arbitrary sequence of $2D$ sets from $\mathcal{C}$, say $(H_0,\ldots,H_{2D-1})$. The nodes of the graph are the elements of these sets (which are integers),
and this induces a natural labeling of the nodes. Next,
for each $j \in \{0,\ldots,2D-1\}$, add all edges between pairs of elements of the set $H_j$, as well as all edges between every element in $H_j$ and 
every element in $H_{j+1}$ (where the indices are taken modulo $2D$). 
In other words, each graph in $\mathcal{G}$ is a ``fat ring'', as illustrated in Figure \ref{fatring}. 
We uniquely identify each graph in $\mathcal{G}$ by its sequence of sets $(H_0,\ldots,H_{2D-1})$, where the node $x$ with the smallest label belongs to the set $H_0$,
 and $H_1$ contains the smallest neighbour of $x$ outside of $H_0$. The size of $\cal G$ is calculated in the following claim.
 
\begin{figure}[!ht]
\begin{center}
\includegraphics[scale=1]{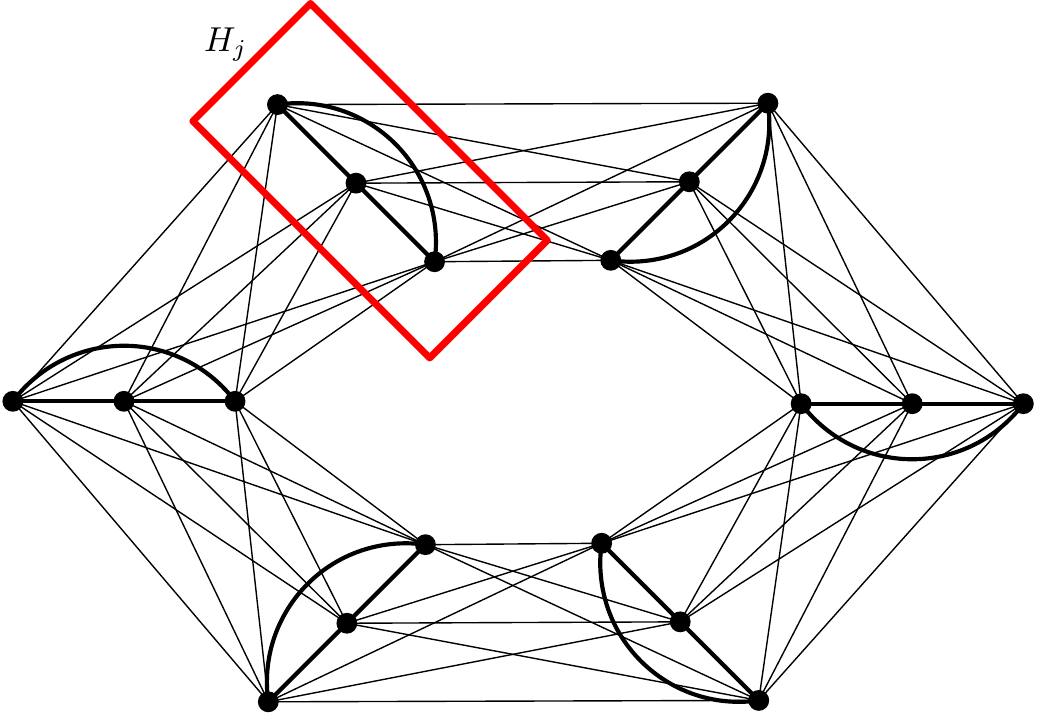}
\end{center}
\caption{A ``fat ring'' with $D=3$ and $n=18$.}
\label{fatring}
\end{figure}

\vspace{3mm}\noindent{\bf Claim 1 } 
 $|{\cal G}|=\binom{n^3}{2D}(2D-1)!/2$.
\vspace{3mm}

To prove the claim, first note that the number of sequences $(H_0,\ldots,H_{2D-1})$ consisting of $2D$ distinct sets from $\mathcal{C}$ is $\binom{n^3}{2D}(2D)!$. To count the number of such sequences belonging to ${\cal G}$, we first divide this integer by $2D$ to eliminate those sequences in which $H_0$ does not contain the smallest label. Then, we divide the result by 2 to eliminate those sequences in which the labels in $H_{2D-1}$ are smaller than those in $H_1$. This completes the proof of the claim.

Next, let $b$ be the maximum number of advice bits provided to ${\cal A}$, taken over all graphs in $\mathcal{G}$. By the Pigeonhole Principle, there exists a family $S$ of at least $\frac{|{\cal G}|}{2^b}$ graphs in $\mathcal{G}$ such that the algorithm receives the same advice string when executed on each graph in $S$.
The following claim will be used to find an upper bound on the size of $S$. 

\vspace{3mm}\noindent{\bf Claim 2 } 
\textit{Consider two graphs from $S$, say $S_1 = (A_0,\ldots,A_{2D-1})$ and $S_2 = (B_0,\ldots,B_{2D-1})$. Suppose that, for some $i \in \{0,\ldots,2D-1\}$, algorithm $\mathcal{A}$ elects a node from set $A_i$ when executed on $S_1$ and elects a node from set $B_i$ when executed on $S_2$. If $A_j = B_j$ for each $j \neq i$, then $A_i = B_i$.}
\vspace{3mm}

To prove the claim, let $k = i + D\ (\!\!\!\mod 2D$).
Since $A_j = B_j$ for all $j \neq i$, it follows that each node $v \in A_k$ is also an element of $B_k$, and, moreover, for each such $v$,
knowledge $K(D-1,v)$ in $S_1$ is equal to knowledge $K(D-1,v)$ in $S_2$.
The nodes of $A_k$ output some label $\ell$ at the end of the execution of $\mathcal{A}$ on $S_1$.
Since we assumed that every execution of $\cal A$ halts within $D-1$ rounds, and the same advice is given for $S_1$ and $S_2$, 
the nodes of $B_k$ also output label $\ell$ at the end of execution of $\mathcal{A}$ on $S_2$. Since $\mathcal{A}$ elects a node from $A_i$ when executed on $S_1$ and elects a node from $B_i$ when executed on $S_2$, it follows that $\ell$ is a label that appears in both $A_i$ and $B_i$. As the sets in $\mathcal{C}$ are pairwise disjoint, it follows that $A_i = B_i$. This concludes the proof of the claim.

Using Claim 2, we now obtain an upper bound on the size of $S$. In particular, for each $i \in \{0,\ldots,2D-1\}$, consider the subfamily of graphs $(H_0,\ldots,H_{2D-1})$ in $S$ such that algorithm $\mathcal{A}$ elects a node from $H_i$. By the claim, for each choice of the $2D-1$ sets $H_0,\ldots,H_{i-1},H_{i+1},\ldots,H_{2D-1}$, there is exactly one set $H_i$ such that $(H_0,\ldots,H_{2D-1})$ belongs to $S$. The number of such choices is bounded above by $(n^3)(n^3-1) \cdots (n^3-2D+2) = \binom{n^3}{2D-1}(2D-1)!$. Since this is true for all $2D$ possible values for $i$, we get that $|S| \leq 2D\binom{n^3}{2D-1}(2D-1)!$. Comparing this upper bound to our lower bound on $|S|$, it follows from Claim 1 that $2D\binom{n^3}{2D-1}(2D-1)! \geq \frac{\binom{n^3}{2D}(2D-1)!/2}{2^b}$. Re-arranging this inequality, we get that $2^b \geq \frac{\binom{n^3}{2D}}{4D\binom{n^3}{2D-1}} = \frac{n^3-2D+1}{8D^2} \in \Omega(n)$, and hence $b\in \Omega(\log n)$.
\end{proof}

Hence we have the following corollary.

\begin{corollary}
The optimal size of advice to complete election in any graph in time less than its diameter is $\Theta (\log n)$.
\end{corollary}

\section{Selection}

In this section, we study the selection problem for the class of rings. It turns out that significant differences between election and selection can be exhibited
already for this class. As in the case of election, and for the same reasons, selection in time larger than the diameter can be accomplished without any advice. Hence, in the rest of the section, we consider selection in time at most equal to the diameter. 
For any ring $R$ and any selection algorithm $\algname$, denote by $\advice{R}$ the advice string provided to all nodes in $R$ when they execute algorithm $\algname$. Denote by $\selected{R}$ the node that outputs 1 in the execution of algorithm $\algname$ on ring $R$.
We first look at selection algorithms working in time at most $\alpha \cdot \diam{R}$, for any ring $R$ and any constant $0<\alpha \leq 1$.

\subsection{Selection in time linear in the diameter}

We start with the lower bound on the size of advice needed by any selection algorithm working in time equal to the diameter of the ring.
This lower bound shows that, for any positive integer $D$, there exists a ring with diameter at most $D$ for which such an algorithm requires 
$\Omega(\log\log D)$ bits of advice. Of course, this implies the same lower bound on the size of advice for selection in any smaller time.

The following theorem provides our first lower bound on the size of advice for selection. 

\begin{theorem}\label{lb1}
Consider any selection algorithm $\algname$ such that, for any ring $R$, algorithm $\algname$ halts within $\diam{R}$ rounds. For every positive integer $D$, there exists a ring $R$ of diameter at most $D$ for which algorithm $\algname$ requires advice of size $\Omega(\log\log D)$.
\end{theorem}
\begin{proof}
At a high level we consider a ``rings-into-bins'' problem, in which each bin represents a distinct advice string. We recursively construct sets of rings, such that 
the rings constructed at a given recursion level cannot be put into the same bin as previously-constructed rings. We continue the construction long enough to run out of bins. With few bins, the number of levels of recursion is sufficiently small to keep the diameters of the constructed rings bounded by $D$.

The following claim will be used to show that a particular ring that we construct will cause algorithm $\algname$ to fail.
We will use the following {\em chopping} operation in our constructions.
For any selection algorithm $\algname$ and any ring $R$ of odd size, we define the \emph{chop} of $R$, denoted by $\chop{R}$, to be the path obtained from $R$ by removing the edge between the two nodes at distance $\diam{R}$ from $\selected{R}$ (see Figure \ref{chopandglue}(a)).

\vspace{3mm}\noindent{\bf Claim 1 } 
\textit{Consider any selection algorithm $\algname$ such that, for any ring $R$, algorithm $\algname$ halts within $\diam{R}$ rounds. Consider two disjoint rings $R_1, R_2$ of odd size such that $\advice{R_1} = \advice{R_2}$. For any ring $R_3$ that contains $\chop{R_1}$ and $\chop{R_2}$ as subgraphs such that $\advice{R_3} = \advice{R_1} = \advice{R_2}$, two distinct nodes in $R_3$ output 1 when executing $\algname$.}
\vspace{3mm}

In order to prove the claim, 
consider the execution of $\algname$ by the nodes of $R_3$. Using the definition of $\chop{R_1}$, it can be shown that knowledge $K(\diam{R_1},\selected{R_1})$ in $R_3$ is equal to knowledge $K(\diam{R_1},\selected{R_1})$ in $R_1$. Therefore, the execution of $\algname$ at node $\selected{R_1}$ in $R_3$ halts in round $\diam{R_1}$, and $\selected{R_1}$ outputs 1. Similarly, the execution of $\algname$ at node $\selected{R_2}$ in $R_3$ halts in round $\diam{R_2}$, and $\selected{R_2}$ outputs 1. Since $R_1$ and $R_2$ are disjoint, we have that $\selected{R_1} \neq \selected{R_2}$. Hence, two distinct nodes in $R_3$ output 1 when executing $\algname$.
This proves the claim.

It is enough to prove the theorem for sufficiently large $D$.
Fix any integer $D \geq 2^{27}$ and set the label space to be $\{1,\ldots,2D\}$. To obtain a contradiction, assume that $A < \sqrt{\log D}$ different advice strings suffice. 

Form a family ${\cal T}$ of $2^{A}A!$ disjoint sequences of integers, each of size 3. More specifically, let ${\cal T} = \{(3i,3i+1,3i+2)\ |\ i \in \{0,\ldots,2^{A}A!-1\}\}$. Let ${\cal E} = \{3(2^{A}A!)+1, 3(2^{A}A!)+2, \ldots, 4(2^{A}A!)\}$. Note that ${\cal E}$ is a set of integer labels, each of which is larger than all labels that belong to sequences in ${\cal T}$. To verify that we have enough labels to define these sets, note that the largest label $\ell = 4(2^{A}A!) < 4(2^{\sqrt{\log D}}[\sqrt{\log D}]!)$. Using the inequality $\log n! \leq (n+1)\log(n+1) + 1$, we get that $\log \ell< 2 + \sqrt{\log D} + (\sqrt{\log D}+1)\log(\sqrt{\log D}+1) + 1 = 3 + \sqrt{\log D} + \sqrt{\log D}\log(\sqrt{\log D}+1) + \log(\sqrt{\log D}+1)$. When $D \geq 2^{27}$, one can verify that $\log(\sqrt{\log D}+1) \leq \sqrt[3]{\log D}$, and $3 + \sqrt{\log D} + \sqrt[6]{(\log D)^5} + \sqrt[3]{\log D} < \log D$, so $\log \ell < \log D$. It follows that $\ell < D <L$. 

Next, we construct a special family ${\cal R}$ of rings of diameter at most $D$. We will add rings to ${\cal R}$ by following a procedure that we will describe shortly. Each new ring that we add to ${\cal R}$ will be the result of at most one `gluing' operation, denoted by $\glue{R_1}{R_2}$, that takes two disjoint odd-sized rings $R_1,R_2$ and forms a new odd-sized ring $R_3$. More specifically, the $i^{th}$ gluing operation takes two disjoint odd-sized rings $R_1$ and $R_2$ and forms the new ring
defined as follows: construct paths $\chop{R_1}$ and $\chop{R_2}$,  add a new edge between the leaf of $\chop{R_1}$ with smaller label and the leaf of $\chop{R_2}$ with smaller label, add a new node
with label $3(2^{A}A!)+i$, and add edges from this new node to the two remaining leaves.  
The gluing operation is illustrated in Figure \ref{chopandglue}(b). The additional node is introduced so that the resulting ring has an odd number of nodes. Note that this additional node's label comes from ${\cal E}$, which ensures that the new ring formed by a gluing operation does not contain duplicate labels. Further, note that, due to the dependence of the additional node's label on $i$, no two gluing operations introduce additional nodes with the same label.

\begin{figure}[!ht]
\begin{center}
\includegraphics[scale=0.8]{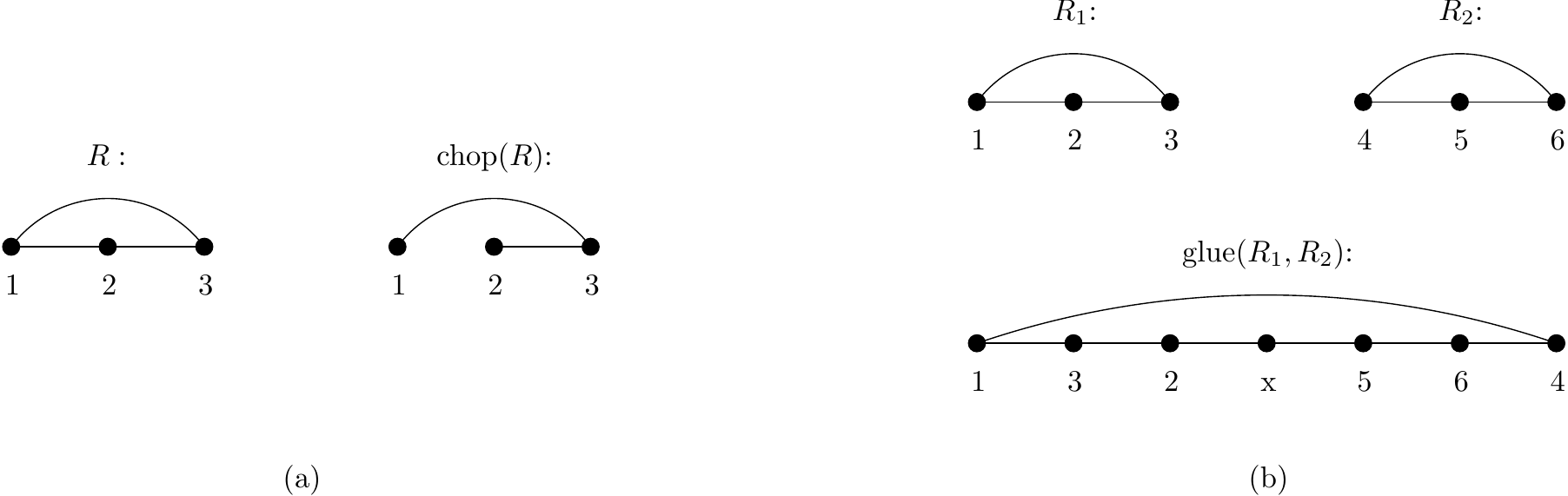}
\end{center}
\caption{(a) The chop operation on a ring of size 3: the edge between the two nodes at distance $\diam{R}=1$ from node $\selected{R} = 3$ is removed. (b) The glue operation on two rings of size 3: concatenate paths $\chop{R_1}$, $\chop{R_2}$ and add an additional node $x$ with a label from ${\cal E}$.}
\label{chopandglue}
\end{figure}

We now describe the procedure for adding rings to ${\cal R}$. In stage 1, we consider the set $G_1$ of 3-cliques obtained from each sequence in ${\cal T}$ by identifying each integer with a node and adding all edges between them. We take a subset $H_1 \subseteq G_1$ of size at least $2^{A}(A-1)!$ such that the same advice string is provided to the algorithm for each ring in $H_1$. The rings in $H_1$ are added to ${\cal R}$, which concludes stage 1. In each stage $j \in \{2,\ldots,A\}$, we consider the set $H_{j-1}$ of rings that were added to ${\cal R}$ in stage $j-1$. The elements of $H_{j-1}$ are partitioned into $|H_{j-1}|/2$ pairs of rings in an arbitrary way. For each such pair $R_1,R_2$, we perform $\glue{R_1}{R_2}$. Define $G_j$ to be the set of all of the resulting rings.   We take a subset $H_j \subseteq G_j$ of size at least $|G_j|/(A-j+1)$ such that the same advice string is provided to the algorithm for each ring in $H_j$. The rings in $H_j$ are added to ${\cal R}$, which concludes stage $j$. This concludes the construction of ${\cal R}$.

It is not immediately clear that this construction can always be carried out. In particular, in order to define $H_j$ in each stage $j \in \{1,\ldots,A\}$,
there must exist $|G_j|/(A-j+1)$ rings in $G_j$ such that the same advice string is provided to the algorithm for each. To prove this fact, and to obtain the desired contradiction to prove the theorem, we will use the following two claims. 

\vspace{3mm}\noindent{\bf Claim 2 } 
\textit{Consider any $k \in \{1,\ldots,A\}$ and any ring $R \in H_k$. For every $j < k$, there exists a ring $Q_j \in H_j$ such that $\chop{R}$ contains $\chop{Q_j}$ as a subgraph.}
\vspace{3mm}

To prove the claim, we proceed by induction on $k$. The case where $k=1$ is trivial. Next, assume that, for some $k \in \{1,\ldots,A\}$ and any ring $R \in H_k$, for every $j < k$, there exists a ring $Q_j \in H_j$ such that $\chop{R}$ contains $\chop{Q_j}$ as a subgraph. We now prove that, for any $R \in H_{k+1}$, there exists a ring $Q_k \in H_k$ such that $\chop{R}$ contains $\chop{Q_k}$ as a subgraph. By construction, $R = \glue{R_1}{R_2}$ for some disjoint rings $R_1,R_2 \in H_k$. By the definition of the gluing operation, $R$ contains both $\chop{R_1}$ and $\chop{R_2}$ as disjoint subgraphs. Since $\chop{R}$ has one fewer edge than $R$, at least one of  $\chop{R_1}$ and $\chop{R_2}$ is a subgraph of $\chop{R}$. By induction, this proves Claim 2.

\vspace{3mm}\noindent{\bf Claim 3 } 
\textit{Consider any $j,k \in \{1,\ldots,A\}$ with $j < k$. For any ring $W \in H_k$, there exist disjoint $P_j,Q_j \in H_j$ such that $W$ contains $\chop{P_j}$ and $\chop{Q_j}$ as subgraphs.}
\vspace{3mm}

To prove the claim, note that, by our construction,  $W = \glue{R_1}{R_2}$ for some disjoint $R_1,R_2 \in H_{k-1}$. It follows that $\chop{R_1}$ and $\chop{R_2}$ are disjoint subgraphs of $W$. If $j=k-1$, setting $P_j = R_1$ and $Q_j = R_2$ satisfies the statement of the claim. If $j < k-1$, then, by Claim 2, there exists a ring $P_j \in H_j$ such that $\chop{R_1}$ contains $\chop{P_j}$ as a subgraph. Similarly, there exists a ring $Q_j \in H_j$ such that $\chop{R_2}$ contains $\chop{Q_j}$ as a subgraph. Note that, since $R_1$ and $R_2$ are disjoint, it follows that $\chop{P_j}$ and $\chop{Q_j}$ are disjoint, so $P_j$ and $Q_j$ are disjoint. Thus, $P_j$ and $Q_j$ satisfy the statement of the claim, which concludes its proof.

We now show that, in any fixed stage of the above construction of ${\cal R}$, the advice string that is provided for the rings added to ${\cal R}$ in this stage is different than the advice strings provided for the rings added to ${\cal R}$ in all previous stages.

\vspace{3mm}\noindent{\bf Claim 4 } 
\textit{Consider any $j,k \in \{1,\ldots,A\}$ with $j < k$. For any ring $W_1 \in H_j$ and any ring $W_2 \in H_k$, we have $\advice{W_1} \neq \advice{W_2}$.}
\vspace{3mm}

To prove the claim, notice that, by Claim 3, there exist disjoint $P_j,Q_j \in H_j$ such that $W_2$ contains $\chop{P_j}$ and $\chop{Q_j}$ as subgraphs. Recall that the algorithm is provided the same advice string for all graphs in $H_{j}$, so $\advice{P_j} = \advice{Q_j} = \advice{W_1}$. By Claim 1 and the correctness of $\cal A$, it follows that $\advice{W_2} \neq \advice{P_j}$. Hence, $\advice{W_2} \neq \advice{W_1}$, which completes the proof of Claim 4.

We show that the construction of ${\cal R}$ can always be carried out. 

\vspace{3mm}\noindent{\bf Claim 5 } 
\textit{For all $k \in \{1,\ldots,A\}$, in stage $k$ of the construction, there exist at least $|G_k|/(A-k+1)$ rings in $G_k$ such that the same advice string is provided to the algorithm for each of them.}
\vspace{3mm}

To prove the claim, first note that, for $k=1$, there are $A$ different strings that could be used as advice for rings in $G_1$. So, by the Pigeonhole Principle, there are at least $|G_1|/A$ rings in $G_1$ such that the same advice string is provided to the algorithm for each of them. Next, for any $k \in \{2,\ldots,A\}$, Claim 4 implies that there are $k-1$ strings that are not provided to the algorithm as advice for rings in $G_k$. Namely, there are at most $A-k+1$ different strings used as advice for rings in $G_k$. By the Pigeonhole Principle, there are at least $|G_k|/(A-k+1)$ rings in $G_k$ such that the same advice string is provided to the algorithm for each of them, which proves the claim.

The following claim implies that $|H_j| \geq 2$, for all $j \leq A$. Later, we will use two rings from some $H_j$ to obtain the desired contradiction needed to complete the proof of the theorem. 

\vspace{3mm}\noindent{\bf Claim 6 } 
\textit{For all $j \in \{1,\ldots,A\}$, at the end of stage $j$, $|H_j| \geq 2^{A-j+1}(A-j)!$.}
\vspace{3mm} 

We prove the claim by induction on $j$. When $j=1$, we have $|G_1| = |{\cal T}| = 2^{A}A!$. By Claim 5, there exist at least $|G_1|/A = 2^{A}(A-1)!$ rings in $G_1$ such that the same advice string is provided to the algorithm for each of them. This implies that $|H_1| \geq 2^{A}(A-1)!$, as required. As induction hypothesis, assume that, at the end of some stage $j \in \{1,\ldots,A-1\}$, $|H_j| \geq 2^{A-j+1}(A-j)!$. In stage $j+1$, the set $H_{j}$ is partitioned into pairs and $G_{j+1}$ consists of one (glued) ring for each such pair. Thus, $|G_{j+1}| = 2^{A-j}(A-j)!$. By Claim 5, there exist at least $|G_{j+1}|/(A-j) = 2^{A-j}(A-j-1)!$ rings in $G_{j+1}$ such that the same advice string is provided to the algorithm for each of them. This implies that $|H_{j+1}| \geq 2^{A-j}(A-j-1)!$, as required. This concludes the proof of Claim 6.

Finally, we construct a ring $X$ on which algorithm $\algname$ fails. Note that the rings in ${\cal R}$ all have node labels from the sets ${\cal T}$ and ${\cal E}$, and we proved that the largest integer in these sets is less than $D$. Thus, the rings in ${\cal R}$ all have node labels from  the range $\{1,\ldots, D\}$. To construct $X$, we take any two (disjoint) rings $R_1,R_2 \in H_A$ (which exist by Claim 6) and form a ring $X$ consisting of the concatenation of paths $\chop{R_1}$, $\chop{R_2}$, and the path of $D$ nodes with labels $D+1,\ldots,2D$. Recall that, in the construction, all rings in $H_j$, for any fixed $j \leq A$, get the same advice string.   
By Claim 4, each of the $A$ distinct advice strings is used for rings in some $H_j$. Therefore, there exists a stage $j$ such that the advice provided for all graphs in $H_j$ is the string $\advice{X}$. By Claim 2, there exists a ring $P_j \in H_j$ such that $\chop{R_1}$ contains $\chop{P_j}$ as a subgraph, and, there exists a ring $Q_j \in H_j$ such that $\chop{R_2}$ contains $\chop{Q_j}$. It follows that $X$ contains $\chop{P_j}$ and $\chop{Q_j}$ as subgraphs, and $\advice{X} = \advice{P_j} = \advice{Q_j}$. By Claim 1, when algorithm $\algname$ is executed by the nodes of ring $X$, two distinct nodes output 1, which contradicts the correctness of $\algname$. Note that the size of $X$ is at least $D$ and at most $2D$, which implies that the size of the label space is linear in the size of $X$. The obtained contradiction was due to the assumption that $A < \sqrt{\log D}$. Hence the number $A$ of different advice strings is at least $\sqrt{\log D}$, which implies that the size of advice is
$\Omega (\log\log D)$ for some ring of diameter at most $D$.
\end{proof}

Since imposing less time cannot make the selection task easier, we have the following corollary.

\begin{corollary}\label{cor}
For any constant $\alpha \leq 1$,
consider any selection algorithm $\algname$ such that, for any ring $R$, algorithm $\algname$ halts within $\alpha \cdot \diam{R}$ rounds.  For every positive integer $D$, there exists a ring $R$ of diameter at most $D$ for which algorithm $\algname$ requires advice of size $\Omega(\log\log D)$.
\end{corollary}

We now establish an upper bound on the size of advice that matches the lower bound from Theorem \ref{lb1}.
Let $D$ be any positive integer and let $L$  be a power of 2.  We consider algorithms that solve selection on the class of rings with diameter at most $D$ and labels from $\{1,\ldots,L\}$. Recall that we assume that $L$ is polynomial in the size of the ring, and hence also in $D$.


In order to prove the upper bound, we propose a family of selection algorithms such that, for each fixed $\alpha \in (0,1]$, there is an algorithm in the family that takes $O(\log\log D)$ bits of advice, and, for each ring $R$, halts within $\alpha \cdot \diam{R}$ rounds. 

We start with an informal description of the algorithm and the advice for any fixed $\alpha \in (0,1]$. The algorithm consists of two stages, and the advice consists of two substrings $A_1$ and $A_2$.

For any ring $R$, the substring $A_1$ of the advice is the binary representation of the integer $a_1=\lfloor\log(\diam{R})\rfloor$. The size of this advice is $O(\log\log D)$. Note that $2^{a_1} \leq \diam{R} < 2^{a_1+1}$.

In stage 1 of the algorithm the nodes perform $r = \lfloor\alpha 2^{a_1}\rfloor$ communication rounds, after which each node $v$ has acquired knowledge $K(r,v)$. Next, each node $v$ checks if its own label is the largest of the labels it has seen within $r$ communication rounds, i.e., the largest in the set $\lambda(r,v)$. If not, then $v$ outputs 0 and halts immediately. 
Let $\candidates{R}$ be the set of {\em candidate nodes}, i.e., nodes $v$ whose label is the largest in $\lambda(r,v)$. Clearly, the largest node in $R$ is in $C_R$, and every node knows if it belongs to
 $\candidates{R}$. Nodes in $\candidates{R}$ proceed to the next stage of the algorithm. 

In stage 2 of the algorithm, each node in $\candidates{R}$ determines whether or not it is the largest node in $R$, without using any further communication rounds. This is achieved using $A_2$, the second substring of the advice, which we now describe.

Let $\cal V$ be the family of sets  of labels which contain all labels in $C_R$ and no larger labels.
At a high level, we construct an integer colouring $F$ of the family $\cal V$ such that, for any $V \in \cal V$,  
when the colour $F(V)$ is given as advice to candidate nodes, each of them can determine, without any communication, whether or not it is the candidate node with the largest label.
Call such a colouring {\em discriminatory}. 
Substring $A_2$ of the advice will be $F(V)$ for some $V \in \cal V$ and some discriminatory colouring $F$ of $\cal V$. (We cannot simply use $F(C_R)$ because our colouring $F$ will be defined on sets of fixed size, and sets of candidate nodes for different rings do not have to be of equal size.) 
Using $A_2$, the candidate nodes solve selection among themselves.
This concludes the high-level description of the algorithm.

The main difficulty of the algorithm is finding  a discriminatory colouring $F$. 
For example, bijections are trivially discriminatory, as nodes could deduce the set $V$.
However, we cannot use such a colouring. Indeed, the colouring must use few colours, otherwise the advice would be too large.
We will be able to construct a discriminatory colouring with few colours  using the fact that the number of candidate nodes is bounded by a constant that depends only on $\alpha$, as given in the following lemma. 


\begin{lemma}\label{boundcandidates}
$1 \leq |\candidates{R}| <8/\alpha$.
\end{lemma}
\begin{proof}
First note that $|\candidates{R}| \geq 1$ since the largest node $v$ in $R$ also has the largest label in $\lambda(r,v)$. 
If $\diam{R}<2/\alpha$, then $|C_R|$ is at most the number of nodes in $R$, i.e., at most $2\cdot diam(R) +1$, which is at most $5/\alpha$. Hence we may assume that $\diam{R}\geq 2/\alpha$.
Next, note that there cannot be two candidate nodes $v_1,v_2$ such that the distance between them is at most $r$. Indeed, if $v_1 \in \lambda(r,v_2)$ and $v_2 \in \lambda(r,v_1)$, then the node in $\{v_1,v_2\}$ with smaller label will not be a candidate. Since the number of nodes in $R$ is $2\cdot \diam{R} < 2^{a_1+2}$, it follows that the number of candidate nodes is less than $(2^{a_1+2})/r = (2^{a_1+2})/(\lfloor \alpha 2^{a_1} \rfloor)$. 
Since $diam(R) \geq 2/\alpha$ and $a_1=\lfloor\log(\diam{R})\rfloor$, it can be shown that $(2^{a_1+2})/(\lfloor \alpha 2^{a_1} \rfloor)$ is at most $(2^{a_1+2})/( \alpha 2^{a_1-1})=8/\alpha$. This completes the proof of the lemma.
\end{proof}

We define what it means for a colouring to be \emph{legal}. 
It will be shown that a legal colouring known by all nodes is discriminatory. Let $\maxcandidates = \maxcandidatesval$.

\begin{definition}
Let ${\cal M}_{\maxcandidates}$ denote the set of all $\maxcandidates$-tuples of the form $(\ell_0,\ell_1,\ldots,\ell_{\maxcandidates-1})$ such that $\ell_0,\ell_1,\ldots,\ell_{\maxcandidates-1} \in \{1,\ldots,L\}$ and $\ell_0 > \ell_1 > \cdots > \ell_{\maxcandidates-1}$. (We identify the tuple $(\ell_0,\ell_1,\ldots,\ell_{\maxcandidates-1})$
with the set $\{\ell_0,\ell_1,\ldots,\ell_{\maxcandidates-1}\}$.)

A colouring $F$ of elements of ${\cal M}_{\maxcandidates}$ by integers is \emph{legal} if, for each colour $c$ and each integer $z \in \{1,\ldots,L\}$, either
\begin{enumerate}
\item every $(\ell_0,\ldots,\ell_{\maxcandidates-1}) \in {\cal M}_{\maxcandidates}$ that contains $z$ and is coloured $c$ has $\ell_0 = z$, or,
\item every $(\ell_0,\ldots,\ell_{\maxcandidates-1}) \in {\cal M}_{\maxcandidates}$ that contains $z$ and is coloured $c$ has $\ell_0 \neq z$.
\end{enumerate}
\end{definition}

Informally, a colouring of sets of labels is legal if,  for all sets of a given colour in which a label $z$ appears,
$z$ is either always the largest label or never the largest label. 

Assume that we have a legal colouring $F$ of ${\cal M}_{\maxcandidates}$ that uses $O(\log^{\maxcandidates} D)$ colours, and that each node knows $F$. Using $F$, we provide a complete description of our algorithm with advice of size $O(\log\log D)$ and prove that it is correct. We will then describe such a legal colouring $F$.



\begin{center}
\fbox{
\begin{minipage}{0.8\columnwidth}\small

{\bf Advice Construction}, for fixed $\alpha \in (0,1]$\\

Input: Ring $R$ with diameter at most $D$\\

1: $a_1 := \lfloor \log(\diam{R}) \rfloor$\\
2: $A_1 := \textrm{binary representation of }a_1$\\
3: $\candidates{R} := \{ v\ |\ \textrm{node $v$ has the largest label in $\lambda(\lfloor\alpha 2^{a_1}\rfloor,v)$} \}$\\
4: $(\gamma_0,\ldots,\gamma_{|\candidates{R}|-1}) := \textrm{labels of nodes in $\candidates{R}$, sorted in descending order}$\\
5: $\maxcandidates := \maxcandidatesval$\\
6: $(\ell_0,\ldots,\ell_{\maxcandidates-1}) := \textrm{a decreasing sequence of labels such that $\ell_0 = \gamma_0$}$\\ \hspace*{3.05cm}\textrm{and }$\{\gamma_1,\ldots,\gamma_{|\candidates{R}|-1}\} \subseteq \{\ell_1,\ldots,\ell_{\maxcandidates-1}\}$\\
7: $A_2 := \textrm{binary representation of }F((\ell_0,\ldots,\ell_{\maxcandidates-1}))$\\
8: Advice := $(A_1,A_2)$
\end{minipage}
}
\end{center}

\begin{center}
\fbox{
\begin{minipage}{0.8\columnwidth}\small

{\bf Algorithm} {\tt Select} at node $v$ with label $\ell$, for fixed $\alpha \in (0,1]$\\

Input: Advice $(A_1,A_2)$\\

1:\phantom{0} $a_1 := \textrm{integer value represented by $A_1$}$\\
2:\phantom{0}  Acquire knowledge $K(\lfloor\alpha 2^{a_1}\rfloor,v)$ using $\lfloor\alpha 2^{a_1}\rfloor$ communication rounds\\
3:\phantom{0}  \textbf{if } $\ell$ is not the largest label in $\lambda(\lfloor\alpha 2^{a_1}\rfloor,v)$ \textbf{ then}\\
4:\phantom{0}  \hspace*{0.7cm} Output 0 and halt.\\
5:\phantom{0}  $a_2 := \textrm{integer value represented by $A_2$}$\\
6:\phantom{0}  $\maxcandidates := \maxcandidatesval$\\
7:\phantom{0}  $Inv := \{\textrm{all $\maxcandidates$-tuples $X$ such that $F(X) = a_2$}\}$\\
8:\phantom{0}  \textbf{if } there is a tuple in $Inv$ with first entry equal to $\ell$ \textbf{ then}\\
9:\phantom{0}  \hspace*{0.7cm} Output 1 and halt.\\
10: \textbf{else}\\
11: \hspace*{0.7cm} Output 0 and halt.
\end{minipage}
}
\end{center}

In order to complete the description, it remains to construct a legal colouring $F$
that uses $O(log^{\maxcandidates} L)$ colours. Note that, since $L$ is polynomial in $D$, that the number of colours is indeed $O(log^{\maxcandidates} D)$.

Consider the following mapping $g: {\cal M}_{\maxcandidates} \longrightarrow \{0,\dots ,  \log L \}^{\maxcandidates-1}$ that maps each $(\ell_0,\ell_1,\ldots,\ell_{\maxcandidates-1}) \in {\cal M}_{\maxcandidates}$ to a $(\maxcandidates-1)$-tuple $(a_1,\ldots,a_{\maxcandidates-1})$. For each $i \in \{1,\ldots,\maxcandidates-1\}$, set $a_i$ to be the largest integer $j$ such that there exists an integer in the range $[\ell_0+1,\ldots,\ell_i]$ that is divisible by $2^j$. Next, take any bijection $f$ between $\{0,\dots ,  \log L \}^{\maxcandidates-1}$ and  the set $\{0,\dots, ((\log L)+1)^{\maxcandidates-1}-1\}$. 
Define the colouring $F$ of  ${\cal M}_{\maxcandidates}$ as the composition of  $g$ and $f$. Note that the colouring $F$ uses at most $((\log L)+1)^{\maxcandidates-1}$ colours. 

\begin{lemma}\label{isLegal}
$F$ is a legal colouring of the elements of ${\cal M}_{\maxcandidates}$.
\end{lemma}
\begin{proof}
Consider an arbitrary colour $c \in \{1,\dots, ((\log L)+1)^{\maxcandidates-1}\}$ and any integer $z \in \{1,\ldots,L\}$. To obtain a contradiction, assume that: 
\begin{itemize}
\item there exists an element of ${\cal M}_{\maxcandidates}$, say $(x_0,\ldots,x_{\maxcandidates-1})$, that contains $z$, is coloured $c$, and has $x_0 \neq z$, and,
\item there exists an element of ${\cal M}_{\maxcandidates}$, say $(y_0,\ldots,y_{\maxcandidates-1})$, that contains $z$, is coloured $c$, and has $y_0 = z$.
\end{itemize}
Let $(a_1,\ldots,a_{\maxcandidates-1})=f^{-1}(c)$.
Since $(x_0,\ldots,x_{\maxcandidates-1})$ contains $z$ and $x_0 \neq z$, it follows that $x_i = z$ for some $i \in \{1,\ldots,\maxcandidates-1\}$. By the definition of $g$, we know that the range $[x_0+1,\ldots,z]$ contains an integer, say $z_1$, that is divisible by $2^{a_i}$, and we know that no integer in this range is divisible by $2^{a_i+1}$. Moreover, since $y_0 = z$, we know that the range $[z+1,\ldots,y_i]$ contains an integer, say $z_2$, that is divisible by $2^{a_i}$, and we know that no integer in this range is divisible by $2^{a_1+1}$. It follows that there are two distinct integers $z_1 < z_2$ in the range $[x_0+1,\ldots,y_i]$ that are divisible by $2^{a_1}$ and that no integer in the range $[x_0+1,\ldots,y_i]$ is divisible by $2^{a_1+1}$. Since $z_1+2^{a_1}$ is the smallest integer greater than $z_1$ that is divisible by $2^{a_1}$, it follows that $z_1+2^{a_1} \in [z_1,\ldots,z_2] \in [x_0+1,\ldots,y_i]$, so $z_1+2^{a_1}$ is not divisible by $2^{a_1+1}$. Note that, for some positive integer $q$, we can write $z_1 = 2^{a_1}q$ and $z_1 + 2^{a_1} = 2^{a_1}(q+1)$. Since neither $z_1$ nor $z_1 + 2^{a_1}$ is divisible by $2^{a_1+1}$, it follows that both $q$ and $q+1$ must be odd, a contradiction.
\end{proof}

\begin{theorem}\label{ub}
Consider any fixed $\alpha \in (0,1]$ and any positive integer $D$. For any ring $R$ with diameter at most $D$,
Algorithm {\tt Select} solves selection in the ring $R$ 
in time $\alpha\cdot\diam{R}$ and with advice of size $O(\log\log D)$.
\end{theorem}

\begin{proof}
Notice that Algorithm {\tt Select} halts within  $\alpha\cdot\diam{R}$ communication rounds. Indeed, 
by line 2 of the algorithm, every node uses exactly $\lfloor \alpha 2^{a_1} \rfloor$ communication rounds, 
and $2^{a_1} \leq \diam{R}$. Next, note that since $A_1$ is the binary representation of $\lfloor \log(\diam{R}) \rfloor \leq \log D$, the length of $A_1$ is $O(\log\log D)$. Further, recall that $A_2$ is the binary representation of a colour assigned by $F$, which uses $O(\log^{\maxcandidates} D)$ colours where $\maxcandidates \in O(1)$. Thus, the length of $A_2$ is $O(\log\log D)$, and hence the size of advice is  $O(\log\log D)$.

Finally, we prove the correctness of the algorithm for an arbitrary ring $R$. First, note that the construction of the advice string can indeed be carried out. In particular, at line 6 in the Advice Construction, the tuple $(\ell_0,\ldots,\ell_{\maxcandidates-1})$ exists since, by Lemma \ref{boundcandidates}, the number of candidates (and, hence, the value of $|\candidates{R}|$) is bounded above by $\maxcandidates$. Next, recall that $\cal V$ is the family of sets  of labels which contain all labels in $C_R$ and no larger labels. The following claim shows that the colouring $F$ is discriminatory.
 



\vspace{3mm}\noindent{\bf Claim 1 } 
\textit{For any $V= (\ell_0,\ldots,\ell_{\maxcandidates-1}) \in \cal V$, let $a_2 = F(V)$. When lines 6 -- 11 of Algorithm {\tt Select} are executed, the largest node in $\candidates{R}$ outputs 1, and all other nodes in $\candidates{R}$ output 0. This proves that colouring $F$ is discriminatory.}
\vspace{3mm}

To prove the claim, consider the largest node $w \in R$. We first show that $\ell_0$ is equal to $w$'s label. Indeed, since $w$ has the largest label in $\lambda(\lfloor\alpha 2^{a_1}\rfloor,w)$, it follows that $w \in \candidates{R}$. Since $V$ does not contain labels larger than those in $C_R$, $\ell_0$ is equal to $w$'s label.

Next, note that, by line 7, $V \in Inv$. Since $\ell _0$ is equal to $w$'s label, the condition of the \textbf{if} statement at line 8 evaluates to true, so $w$ outputs 1. Next, consider any node $v \in C_R$ that is not the largest.
Since $\ell _0$ is not equal to $v$'s label and $F$ is a legal colouring, it follows that no tuple that is coloured $a_2$ has $v$'s label as its first entry. In other words, no tuple in $Inv$ has $v$'s label as its first entry, so the \textbf{if} statement at line 8 evaluates to false, and $v$ outputs 0. This completes the proof of the claim.

We have shown that, if the nodes in $\candidates{R}$ are provided as advice the value of $F(V)$ for any $V \in {\cal V}$, then lines 6 -- 11 of Algorithm {\tt Select} solve selection among the nodes in $\candidates{R}$. It remains to show that the advice substring $A_2$ created in the Advice Construction is indeed the binary representation of such an $F(V)$.

\vspace{3mm}\noindent{\bf Claim 2 } 
\textit{In Advice Construction, $(\ell_0,\ldots,\ell_{\beta-1}) \in {\cal V}$.}
\vspace{3mm}

To prove the claim, recall that, by line 4 of the Advice Construction, $(\gamma_0,\ldots,\gamma_{|\candidates{R}|-1})$ is the decreasing sequence of labels of nodes in $\candidates{R}$. Consider line 6 of the Advice Construction, and note that the largest element $\ell_0$ of the tuple $(\ell_0,\ldots,\ell_{\beta-1})$ is equal to $\gamma_0$, which proves that the labels $\ell_0,\ldots,\ell_{\beta-1}$ are no larger than those of nodes in $\candidates{R}$. Further, since $\ell_0 = \gamma_0$ and $\{\gamma_1,\ldots,\gamma_{|\candidates{R}|-1}\} \subseteq \{\ell_1,\ldots,\ell_{\maxcandidates-1}\}$, it follows that $(\ell_0,\ldots,\ell_{\beta-1})$ contains all of the labels of nodes in $C_R$. Therefore, $(\ell_0,\ldots,\ell_{\beta-1}) \in {\cal V}$, as claimed.

We can now conclude that Algorithm {\tt Select} solves selection in the entire ring. Indeed, in lines 1 -- 4, each node that learns about the existence of a node with a larger label than itself outputs 0 and halts. The remaining nodes, namely those in $\candidates{R}$ (which necessarily includes the largest node), proceed to lines 5 -- 11. From line 7 of the Advice Construction and by Claim 2, line 5 of Algorithm {\tt Select} assigns to $a_2$ the value of $F(V)$ for some $V \in {\cal V}$. Finally, Claim 1 shows that the largest node outputs 1 and all other nodes output 0.
\end{proof}

Corollary \ref{cor} and Theorem \ref{ub} imply the following tight bound on the size of advice.

\begin{corollary}
The optimal size of advice to complete selection in any ring in time linear in its diameter (and not exceeding it) is $\Theta(\log \log D)$ .
\end{corollary}

\subsection{Selection in time $D^{\epsilon}$ where $0 \leq\epsilon < 1$}

We now turn attention to very fast selection for rings of diameter $D$. Since every such ring has size $n$ linear in $D$, and since $L$ is polynomial in $n$, selection (and even election) can be 
accomplished using $O(\log D)$ bits of advice without any communication by providing the largest label as advice.
It turns out that, even when the available time is $D^{\epsilon}$, for any constant $0 \leq\epsilon < 1$, this size of advice is necessary.  Compared to Theorem \ref{ub},
this shows that
selection in time $D^{\epsilon}$ where $0 \leq\epsilon < 1$, requires exponentially more advice
than selection in time $\Theta(D)$.

\begin{theorem}
For any constant $0\leq \epsilon < 1$,  any selection algorithm $\cal A$
for rings of diameter $D$ that works in time at most $D^{\epsilon}$ requires $\Omega(\log D)$ bits of advice. 
\end{theorem}
\begin{proof}
It is enough to consider sufficiently large values of $D$.
Let $U=\{1,\ldots,2D-1\}$, and let $W=\{2D,\ldots,L\}$. We start by defining a special class of rings, and then proceed to show that algorithm $\cal A$ requires $\Omega(\log D)$ bits of advice for this class. For the sake of clarity, we will use the notation $x=\lceil D^{\epsilon} \rceil$ and $y=\lceil D^{1-2\epsilon} \rceil$.

First, we define a set of paths ${\cal P} = \{P_1,\ldots,P_y\}$ as follows. Each path consists of $1+2x$ nodes with labels from $W$, with the middle node of the path having the largest label. More specifically, path $P_i$ is obtained by considering the path of nodes with labels $2Di,\ldots,2Di+2x$, respectively, and reversing the order of the last $x+1$ labels. Formally, for each $i \in \{1,\ldots,y\}$, let $P_i$ be the path of nodes whose labels form the sequence $(2Di,\ldots,2Di+x-1,2Di+2x,2Di+2x-1,\ldots,2Di+x)$.

Next, we define a set $\cal R$ of rings $\{R_1,\ldots,R_y\}$. Each ring $R_i$ will consist of the paths $P_1,\ldots,P_i \in {\cal P}$ along with enough nodes with labels from $U$ to ensure that $R_i$ has size $2D$. In order to define the ring $R_i$, we first construct a path $J_i$ by taking the paths $P_1,\ldots,P_i$, and, for each $j \in \{2,\ldots,i\}$, connecting the last node of path $P_{j-1}$ with the first node of path $P_j$. More specifically, $J_i$ is obtained by taking the union of the paths $P_1,\ldots,P_i$, and, for each $j \in \{2,\ldots,i\}$, adding an edge between the nodes with labels $2Dj - 2D +x$ and $2Dj$. 
Next, we construct a path $K_i$ whose labels form the sequence $(1,\ldots,2D-|J_i|)$. 
Since the number of nodes in $J_i$ is at most $y(1+2x)$, which, for sufficiently large $D$, is strictly less than $2D$, it follows that path $K_i$ has at least one node.
Finally, the ring $R_i$ is obtained by joining the paths $J_i$ and $K_i$. More specifically, $R_i$ is obtained by taking the union of the paths $J_i$ and $K_i$ and adding the edges $\{1,2Di+x\}$ and $\{2D-|J_i|,2D\}$. The paths and rings constructed above are illustrated in Figure \ref{pathsrings}.

\begin{figure}[!ht]
\begin{center}
\includegraphics[scale=1]{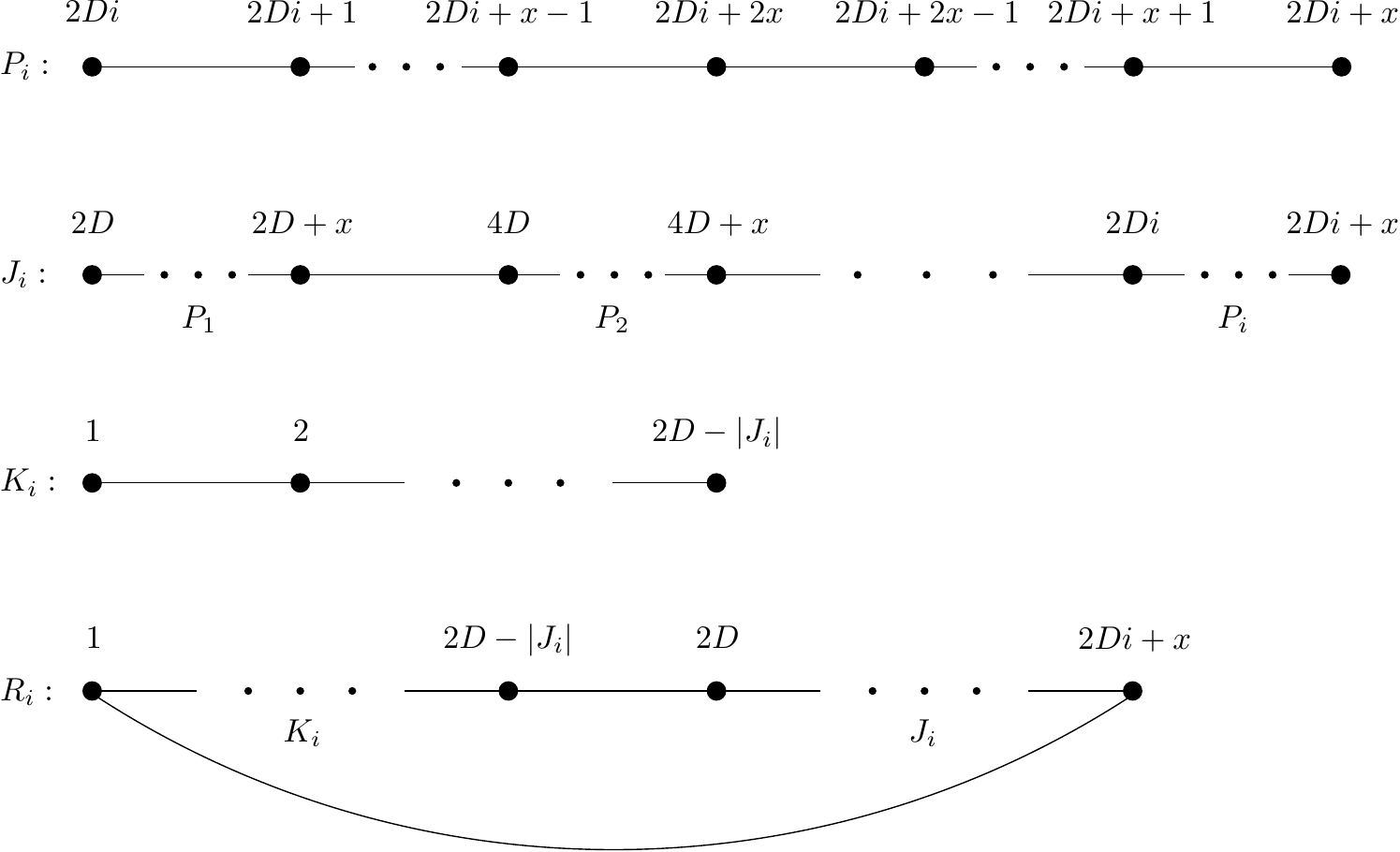}
\end{center}
\caption{Graphs used in the construction of $R_i$}
\label{pathsrings}
\end{figure}

The following claim asserts that, in any two rings in $\cal R$ that contain path $P_i$ as a subgraph, the middle node of this path in both rings acquires the same knowledge 
when executing algorithm $\cal A$.

\vspace{3mm}\noindent{\bf Claim 1 } 
\textit{Consider any $i \in \{1,\ldots,y\}$ and any $j \in \{i,\ldots,y\}$. Let $v_i$ be the node in $R_i$ with label $2Di+2x$ and let $v_j$ be the node in $R_j$ with label $2Di+2x$. Then, knowledge $K(x,v_i)$ in $R_i$ is equal to knowledge $K(x,v_j)$ in $R_j$.}
\vspace{3mm}

To prove the claim, note that, in both $R_i$ and $R_j$, the node with label $2Di+2x$ is the middle node of path $P_i$. Since the path $P_i$ has length $1+2x$, it follows that, within $x$ communication rounds, node $v_i$ does not learn about any nodes in $R_i$ outside of $P_i$. Similarly, node $v_j$ does not learn about any nodes in $R_j$ outside of $P_i$. This implies that knowledge $K(x,v_i)$ in $R_i$ is equal to knowledge $K(x,v_j)$ in $R_j$, which completes the proof of the claim.

We proceed to prove the theorem by way of contradiction. Assume that the number of bits of advice needed by algorithm ${\cal A}$ for the rings in ${\cal R}$ is less than $(1-2\epsilon)\log{D}$. It follows that the number of distinct advice strings that are provided to algorithm $\cal A$ for the rings in $\cal R$ is strictly less than $D^{1-2\epsilon} \leq y$. However, since the class $\cal R$ consists of $y$ rings (where $y \geq 2$ for sufficiently large $D$) this means that there exist two rings in $\cal R$, say $R_a, R_b$ with $a<b$, such that $\advice{R_a} = \advice{R_b}$.

First, consider the execution of $\cal A$ by the nodes of ring $R_a$. Let $v_a$ be the node in $R_a$ with label $2Da+2x$, and note that $v_a$ is the largest node in $R_a$. It follows that, in this execution, node $v_a$ outputs 1. Next, consider the execution of $\cal A$ by the nodes of ring $R_b$. Let $v_b$ be the node in $R_b$ with the label $2Da+2x$. By Claim 1, $K(x,v_a)$ in $R_a$ is equal to $K(x,v_b)$ in $R_b$. Moreover, 
algorithm $\cal A$ halts within $x$ rounds and $\advice{R_a} = \advice{R_b}$. It follows that the execution of $\cal A$ by $v_a$ in $R_a$ is identical to the execution of $\cal A$ by $v_b$ in $R_b$. Hence, $v_b$ outputs 1 in the execution of $\cal A$ in $R_b$. However, node $v_b$ is not the largest in $R_b$. Indeed, $v_b$'s label is $2Da+2x$, whereas path $P_b$ (and, hence, ring $R_b$) contains a node labeled $2Db+2x>2Da+2x$. This contradicts the correctness of~$\cal A$.
\end{proof}

\begin{corollary}
For any constant $0 \leq\epsilon < 1$,
the optimal size of advice to complete selection in any ring of diameter at most $D$ in time $D^{\epsilon}$  is $\Theta(\log D)$.
\end{corollary}

\section{Conclusion}

We established tradeoffs between the time of choosing the largest node in a network and the amount of {\em a priori} information (advice) needed
to accomplish two variations of this task: election and selection. For the election problem, the tradeoff is complete and tight up to multiplicative constants
in the advice size.
Moreover, it holds for the class of arbitrary connected graphs.
For selection, our results are for the class of rings and a small gap remains in the picture. 
We proved that in rings with diameter $diam$ at most $D$, 
the optimal size of advice is $\Theta(\log D)$ if time is
$O(diam^{\epsilon})$ for any $\epsilon <1$, and that it is $\Theta(\log \log D)$ if time is at most $diam$.
Hence, the first open problem is to establish the optimal size of advice to perform selection for rings when the time is in the small remaining gap, for example, in time $\Theta(diam/\log diam)$.  Another problem is to extend the tradeoff obtained for selection in rings to the class
of arbitrary connected graphs. In particular, it would be interesting to investigate whether the optimal advice needed to perform fast selection in graphs of size much larger
than the diameter depends on their size (like in the case of election) or on their diameter.

As noted in the introduction, all known leader election algorithms in labeled networks choose as leader either the node with the largest label or that with 
the smallest label. It is worth noting that the situation is not always symmetric here. For example, the Time Slice algorithm for leader election \cite{Ly} (which elects the
leader using exactly $n$ messages in $n$-node rings, at the expense of possibly huge time) finds the node
with the smallest label, and it does not seem to be possible to convert it directly to finding the node with the largest label. (Of course it is possible to first
find the node with the smallest label and then to use this leader to find the node with the largest label, but this takes additional time and communication.)
In our case, however, the situation is completely symmetric with respect to the order of labels: our results hold without change, if finding the largest node is replaced by finding the smallest. It is an open question whether they also remain valid if electing the largest node is replaced by general leader election (i.e., the task
in which a single {\em arbitrary} node becomes the leader, and all other nodes become non-leaders and also learn the identity of the leader) and if selecting the largest
node is replaced by general  leader selection  (i.e., the task
in which a single {\em arbitrary} node becomes the leader, and all other nodes become non-leaders). 
Our lower bound for election in time at most the diameter and our algorithms for selection use the ordering in an essential way, and it is not clear if smaller advice would be sufficient to elect or select some arbitrary node in a given time.

\pagebreak

\bibliographystyle{plain}


\end{document}